\documentclass[a4paper,11pt]{article}

\usepackage{fullpage}
\usepackage{amsmath}
\usepackage{amsthm}
\usepackage{float}
\usepackage{enumitem}
\usepackage{graphicx}
\usepackage[tight]{subfigure}

\newcommand{\exc}[1]{\mathit{excess}(#1)}
\newcommand{\fin}[1]{\mathit{in}(#1)}
\newcommand{\M}{\mathcal{M}}

\newtheorem{theorem}{Theorem}[section]
\newtheorem{lemma}[theorem]{Lemma}
\newtheorem{invariant}[theorem]{Invariant}

\floatstyle{ruled}
\newfloat{algorithm}{H}{loa}
\floatname{algorithm}{Algorithm}

\newlist{lines}{enumerate}{3}
\setlist[lines]{nolistsep,label*=\arabic*.}

\title{Multiple-source multiple-sink maximum flow in planar graphs}
\author{Yahav Nussbaum\thanks{The Blavatnik School of Computer Science, Tel Aviv University, 69978 Tel Aviv, Israel.\newline{\small \texttt{yahav.nussbaum@cs.tau.ac.il}}}}
\date{}

\begin{document}

\maketitle

\begin{abstract}
    In this paper we show an $O(n^{3/2} \log^2 n)$ time algorithm for finding a maximum flow in a planar graph with multiple sources and multiple sinks. This is the fastest algorithm whose running time depends only on the number of vertices in the graph. For general (non-planar) graphs the multiple-source multiple-sink version of the maximum flow problem is as difficult as the standard single-source single-sink version. However, the standard reduction does not preserve the planarity of the graph, and it is not known how to generalize existing maximum flow algorithms for planar graphs to the multiple-source multiple-sink maximum flow problem.
\end{abstract}

\section{Introduction}

In the standard \emph{maximum flow} problem, we wish to send the maximum possible amount of flow from a specific vertex designated as a \emph{source} to another vertex designated as a \emph{sink}, this is the \emph{single-source single-sink} maximum flow problem. In a more general formulation of the problem we have many sources and many sinks, this is the \emph{multiple-source multiple-sink} maximum flow problem. Ford and Fulkerson \cite{FF62} suggested a simple reduction from the multiple-source multiple-sink maximum flow problem to the standard single-source single-sink problem. We add to the graph a new vertex $s'$ (called \emph{super-source}) and connect it to all the sources in the graph with arcs of infinite capacity. Similarly we add a new vertex $t'$ (\emph{super-sink}) and connect all the sinks in the graph to it with arcs of infinite capacity. Then we solve a single-source single-sink maximum flow problem with $s'$ as the source and $t'$ as the sink.

In this work, we are interested in flow networks whose underlying graph is \emph{planar}. The fastest known algorithm for the single-source single-sink maximum flow in planar graphs is the algorithm of Borradaile and Klein \cite{BK09}, which runs in $O(n \log n)$ time, where $n$ is the number of vertices in the graph. This algorithm was simplified by Schmidt et al.~\cite{STC09} and by Erickson \cite{E10}. See \cite{BK09} for a survey on previous results for the problem. However, it is not known how to generalize existing planar maximum flow algorithms to the multiple-source multiple-sink problem. The Ford and Fulkerson standard reduction does not preserve the planarity of the graph. Consider a planar graph in which four sources form a clique in the underlying undirected graph. When we add the super-sink $s'$ to the graph we get a complete graph of 5 vertices in the underlying undirected graph, which is not planar. This standard reduction, using an algorithm for maximum flow in general graphs such as the algorithm of Goldberg and Tarjan \cite{GT88}, gives us an $O(n^2 \log n)$ time algorithm for our problem.

In addition to uses in standard applications of planar graph, such as transportation networks, the multiple-source multiple-sink maximum flow problem in planar graphs has many applications of energy functions minimization in computer vision. The maximum flow problem gives us a solution to the related \emph{minimum $s-t$ cut problem} \cite{FF62}. Greig et al.\ \cite{GPS89} were the first to use minimum $s-t$ cut algorithms for such a minimization problem of an energy function, for binary image restoration. Since then, many other similar uses have emerged in computer vision, such as image segmentation, multi-camera stereo vision, object recognition, and others, see \cite{BK04} for a survey. From a graph point of view, the graphs for these applications are all similar. For every pixel in the image, the graph contains a vertex, such that vertices of two neighboring pixels are connected by an edge. In addition, we add a source $s$ and a sink $t$ and connect them to the vertices of the pixels. The pixels define a planar graph, in fact a grid, but the source and the sink destroy the planarity of the graph. However, by breaking $s$ into multiple sources and $t$ into multiple sinks, we can reduce the graph to a multiple-source multiple-sink planar graph. 

Miller and Naor \cite{MN95} studied the problem of maximum flow with multiple sources and multiple sinks in planar graphs. When the \emph{demand} (the required difference between the incoming and outgoing amounts of flow) at each vertex is known, they showed how to reduce this problem to a \emph{circulation} problem, which can be solved in planar graphs in $O(n \log ^2 n / \log \log n)$ time \cite{MW10}. Miller and Naor also showed an $O(n \log^{3/2} n)$ time algorithm for the case where all the sinks and the sources are on the boundary of a single face, and generalized it to an $O(k^2 n^{3/2} \log^2 n)$ time algorithm for the case where the sources and the sinks reside on the boundaries of $k$ different faces. The time bound of the first algorithm can be improved to $O(n \log n)$ using the linear time shortest path algorithm of Henzinger et al.~\cite{HKRS97}, and the time bound of the second algorithm can be improved to $O(k^2 n \log^2 n)$ using the $O(n \log n)$ time single-source single-sink maximum flow algorithm of Borradaile and Klein \cite{BK09}. An additional approach that Miller and Naor show iterate over all the sources and all the sinks, and solve the maximum flow problem from every source to every sink, each time taking the residual network of the previous flow as an input. This gives an $O(|S||T|n \log n)$ time bound, where $|S|$ is the number of sources and $|T|$ is the number of sinks, using the algorithm of Borradaile and Klein \cite{BK09}.

Recently, Borradaile and Wulff-Nilsen \cite{BW10} and Klein and Mozes \cite{KM10} showed two algorithms for maximum flow in planar graphs with multiple sources and a single sink (or equivalently, a single source and multiple sinks) in $O(n^{3/2} \log n)$ time. This leads to an $O(\min\{|S|,|T|\} n^{3/2} \log n)$ time algorithm for our problem.

In this paper, we show an $O(n^{3/2} \log^2 n)$ time algorithm for the multiple-source multiple-sink maximum flow problem in planar graphs. This is the fastest algorithm for the problem whose running time depends only on the value of $n$, and not on the set of sources, the set of sinks, or the capacities of the arcs.

Both algorithms for the multiple-source single-sink version of the problem \cite{BW10,KM10} use \emph{preflows} to find the required flow, we use the similar concept of \emph{pseudoflows} for our algorithm. Pseudoflows were originally used for the minimum cost flow problem \cite{GT90}, Hochbaum \cite{H08} was the first to use pseudoflows for the maximum flow problem. As the two algorithms of \cite{BW10,KM10}, we also decompose the problem into smaller subproblems using \emph{cycle separators} \cite{M86}. Cycle separators were first used for maximum flow in planar graphs by Johnson and Venkatesan \cite{JV83}. We use an approach similar to the one of Miller and Naor \cite{MN95} in our algorithm, initially our flow function $f$ is all zero, then we repeatedly solve some maximum flow problem in the residual network with respect to $f$ and add to $f$ the resulting flow.

\section{Preliminaries}

We consider a simple directed planar graph  $G = (V, E)$, where $V$ is the set of vertices and $E$ is the set of arcs.  We denote the number of vertices by $n$, since  the graph is planar we have $|E| = O(n)$. For $v \in V$, $\fin{v} = \{(u, v) \mid (u, v) \in E\}$ is the set of incoming arcs. We assume that $G$ is given with a fixed planar embedding, in other words it is a \emph{plane graph}. Combinatorial representation of such an embedding can be found in $O(n)$ time \cite{HT76}. We assume that the graph is connected, as we can process every connected component separately.

A \emph{flow network} consists of a graph $G$, a set $S$ of vertices designated as \emph{sources}, a set $T$ of vertices designated as \emph{sinks}, and a \emph{capacity function} that assigns to every arc $e \in E$ a finite capacity $c(e) \geq 0$. A function $f: E \to R$ is a \emph{flow function} if and only if it satisfies the following three constraints:
\begin{align}
    f(e) \leq c(e) \quad &\forall e \in E \enspace & \mbox{(capacity constraint),} \\
    f((u, v)) = -f((v,u)) \quad &\forall (u, v) \in E & \mbox{(antisymmetry constraint),}\\
    \sum_{e \in \fin{v}} f(e) = 0 \quad &\forall v \in V \setminus S \cup T \enspace & \mbox{(flow conservation constraint).}
\end{align}
For the antisymmetry constraint we assume that for every arc $(u, v)$ in $E$, the arc $(v, u)$ is also in $E$, if this is not the case to begin with, then we add the arc with capacity $0$. In the planar embedding of $G$ we identify the arcs $(u, v)$ and $(v, u)$, that is, we embed them as a single \emph{edge}. In other words, an edge is a pair of antiparallel arcs.

The \emph{value} of a flow $f$ is $\sum_{e \in \fin{t},t \in T} f(e)$, the amount of flow that enters the sinks. A \emph{maximum flow} is a flow of maximum value.

Goldberg and Trajan \cite{GT90} define \emph{pseudoflow} by removing the flow conservation constraint, in other words, a pseudoflow is a function $f: E \to R$ that satisfies the capacity constraint and the antisymmetry constraint. The \emph{flow excess} of a vertex $v$ is $\exc{v} = \sum_{e \in \fin{v})} f(e)$. We can view any pseudoflow $f$ as a flow function, if we choose the right sets of sources and sinks. Indeed, if we replace $S$ with $S' = \{v \mid \exc{v} > 0\}$ and $T$ with $T' = \{v \mid \exc{v} < 0\}$ then $f$ is a flow from $S'$ to $T'$ with respect to the graph $G$ and the capacity function $c$.

The \emph{residual capacity} of an arc $e$ with respect to a pseudoflow $f$ is $r_f(e) = c(e) - f(e)$. An arc $e$ is \emph{residual} if $r_f(e) > 0$. A \emph{residual path} is a path of residual arcs. A flow is maximum if and only if it has no residual path from $S$ to $T$ \cite{FF62}. The \emph{residual network} with respect to a pseudoflow $f$ is the flow network $G_f$ on the graph $G$ with sources $S$, sinks $T$, and capacity function $r_f$.

\subsection{Cycle separators}

Assume we assign weights that sum to $1$ to the vertices of $G$. Let $C$ be a cycle in the plane (not a cycle of $G$) that meets $G$ only in its vertices, and let $P$ be the set of vertices in which $G$ and $C$ meet. The cycle $C$ separates $G$ into two edge-disjoint \emph{pieces} -- the subgraph of $G$ inside $C$ including the vertices of $P$, and the subgraph of $G$ outside $C$ including the vertices of $P$. If the total weight of vertices strictly inside $C$, and the total weight of vertices strictly outside $C$ (we do not sum the vertices of $P$ in the total) are both at most $2/3$, then the set of vertices $P$ is a \emph{cycle separator}.

Miller \cite{M86} showed that for any triangulated planar graph there is a cycle separator of $O(\sqrt{n})$ vertices\footnote{Miller \cite{M86} considers graphs that are not 2-connected as a special case which might have a single vertex separator instead of a cycle separator, for our purpose it does not matter.}. Furthermore, such a separator can be found in $O(n)$ time. We can assume without loss of generality that $G$ is triangulated, since otherwise we can triangulate it in linear time with arcs of capacity $0$.

For convenience, we use the term \emph{separator} to describe a set of vertices that separates between two pieces even if the pieces do not satisfy the weight constraint.

\section{Flow summation} \label{sec:sum}

A key technique in our algorithm is \emph{flow summation}, defined as follows. Let $f$ be a pseudoflow with respect to some capacity function $c$, and let $f'$ be a pseudoflow with respect to the residual capacity $r_f$. We let $(f + f')(e) = f(e) + f'(e)$. The function $f + f'$ is a pseudoflow satisfying the capacity constraint with respect to $c$. Furthermore, if $f$ and $f'$ are both flow functions (satisfying flow conservation) from the same set of sources to the same set of sinks, then $f + f'$ is also a flow function from this set of sources to the set of sinks.

We use two algorithms of Miller and Naor \cite{MN95} to obtain a maximum flow through flow summation. In the first algorithm, we apply a single-source single-sink maximum flow algorithm for each pair of source and sink:

\begin{algorithm}
  \begin{lines}
    \item Set $f(e) = 0$ for every arc $e$.
    \item For every $s \in S$ do:
    \begin{lines}
        \item For every $t \in T$ do:
      \begin{lines}
        \item Compute a maximum flow $f'$ from $s$ to $t$ in the residual network $G_f$.
        \item Let $f \leftarrow f + f'$.
      \end {lines}
    \end {lines}
    \item Return $f$.
  \end{lines}
  \caption{\@} \label{alg:iterate}
\end{algorithm}

Algorithm \ref{alg:iterate} is described in Section 7 of \cite{MN95} and also in Theorem 4.1 of \cite{BW10}. The correctness of this algorithm implies the following lemma:

\begin{lemma} \label{lem:subset}
    For a subset of the sources $S' \subseteq S$, let $f$ be a maximum flow from $S'$ to $T$, and let $f'$ be a maximum flow from $S \setminus S'$ to $T$ in $G_f$. Then, $f + f'$ is a maximum flow from $S$ to $T$. Symmetrically, for a subset of the sinks $T' \subseteq T$ with a maximum flow $f$ from $S$ to $T'$ and a maximum flow $f'$ in $G_f$ from $S$ to $T \setminus T'$, the sum $f + f'$ is a maximum flow from $S$ to $T$.
\end{lemma}

The second algorithm of Miller and Naor \cite{MN95} that we use is a recursive algorithm. For $X \subseteq V$ we denote $X \cap S$ by $X_s$, and similarly $X_t$ is $X \cap T$. The algorithm divides the vertices of the graph that are sources or sinks into two sets, $L$ and $R$. After finding a maximum flow $f$ from $L_s$ to $R_t$, the algorithm uses the $\emph{cut}$ with respect to $f$ to decompose the problem into smaller problems. The cut with respect to $f$ contains an arc $(u, v)$ if and only if there is a residual path with respect to $f$ form a vertex of $L_s$ to $u$, but there is no residual path from $L_s$ to $v$. After solving the maximum flow problem in the residual network inside each connected component, and adding the resulting flows to $f$, the algorithm finishes by computing maximum flow from $R_s$ to $L_t$ in the residual network.

\begin{algorithm}
  \begin{lines}
    \item Partition $S \cup T$ into two disjoint sets $L$ and $R$.
    \item Compute a maximum flow $f$ from $L_s$ to $R_t$, and let $C$ be the cut with respect to $f$.
    \item Remove the edges of $C$ from the graph.
    \item For each connected component $K$ do:
    \begin{lines}
      \item Find a maximum flow $f'$ from $K_s$ to $K_t$ in the residual network $G_f$ restricted to $K$.
      \item Let $f \leftarrow f + f'$.
    \end {lines}
    \item Restore the edges of $C$ to the graph
    \item Find a maximum flow $f'$ from $R_s$ to $L_t$ in the residual network $G_f$.
    \item Let $f \leftarrow f + f'$.
    \item Return $f$.
  \end{lines}
  \caption{\@} \label{alg:LR}
\end{algorithm}

The correctness of Algorithm \ref{alg:LR} is proven in Lemma 6.1 of \cite{MN95}. It is described there for a specific choice of $L$ and $R$, but as Miller and Naor note in Section 7 of \cite{MN95}, the proof does not depend on this choice of the two sets. Intuitively, we can compute the maximum flow in each connected component separately because even if we had not removed the edges of $C$ from $G$, then a residual path from a vertex in $K_s$ to a vertex in $K_t$ could not contain an arc of $C$, since these arcs are not residual.

The following lemma allows us to use flow summation while keeping some invariant about non-existence of specific residual paths.

\begin{lemma} \label{lem:keepinv}
    Let $f$ be a pseudoflow such that there is no residual path with respect to $f$ from any vertex in a set $A \subseteq V$ to any vertex in a set $B \subseteq V$, let $f'$ be a flow in $G_f$ from $S$ to $T$. If $S \subseteq A$ or $T \subseteq B$, then there is no residual path from $A$ to $B$ with respect to $f + f'$.
\end{lemma}
\begin{proof}
    We can view the flow $f'$ as the sum of paths that carry flow from $S$ to $T$ and of cycles that carry flow \cite{FF62}. Thus, we can view the summation of $f$ and $f'$ as adding the paths and cycles of flow that compose $f'$ to $f$ one by one in some arbitrary order.

    Assume for contradiction that the claim is not true. Before adding $f'$ to $f$ there was no residual path from $A$ to $B$, then at some point a residual path $Q$ was created from a vertex of $A$ to a vertex of $B$. We choose $Q$ to be the first such residual path created, as we added one by one the paths and cycles which compose $f'$, and let $Q'$ be the path or cycle of $f'$ whose addition made $Q$ residual (if the addition of $Q'$ created more than one residual path from $A$ to $B$, we choose one of them arbitrarily to be $Q$).

    Consider the case where $Q'$ is a path (see Fig.~\ref{fig:keepinv:path}). The paths $Q$ and $Q'$ must share a common vertex $v$. If $S \subseteq A$ then we choose $v$ such that $v$ is the last vertex of $Q'$ in $Q$, and let $Q''$ be the path that begins with the prefix of $Q'$ before $v$ and ends with the suffix of $Q$ that starts at $v$. If $T \subseteq B$ then we choose $v$ such that $v$ is the first vertex of $Q'$ in $Q$, and let $Q''$ be the path that begins with the prefix of $Q$ before $v$ and ends with the suffix of $Q'$ that starts at $v$. In both cases, the path $Q''$ was a residual path from $A$ to $B$ before we added $Q'$ to the flow, contradicting the choice of $Q$.

    In the case where $Q'$ is a cycle (see Fig.~\ref{fig:keepinv:cycle}), let $u$ be the first vertex of $Q$ that is also a vertex of $Q'$ and let $v$ be the last vertex of $Q$ common with $Q'$. Since $Q'$ is a cycle it must contain a path from $u$ to $v$. Let $Q''$ be the path that we obtain by replacing the subpath of $Q$  from $u$ to $v$ with the subpath of $Q'$. Again, the path $Q''$ was a residual path from $A$ to $B$ before we added $Q'$ to the flow, contradicting the choice of $Q$.
\end{proof}

\begin{figure}
    \centering
    \subfigure[$Q'$ is a path carrying flow from $s \in S$ to $t \in T$.]{\quad \includegraphics[scale=0.75]{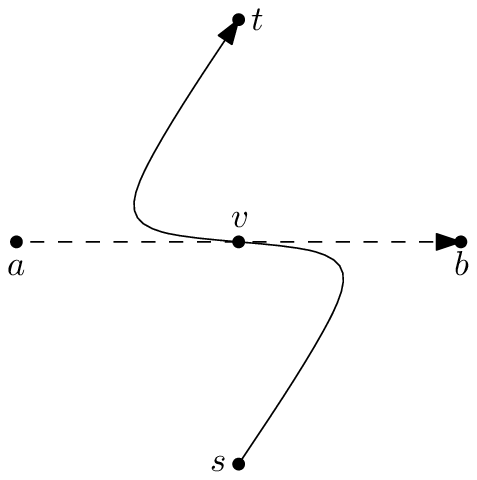} \quad \label{fig:keepinv:path}}
    \qquad
    \subfigure[$Q'$ is a cycle of flow.]{\quad \includegraphics[scale=0.75]{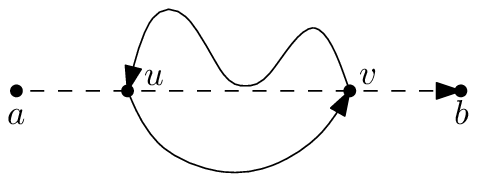} \quad \label{fig:keepinv:cycle}}
    \caption{Proof of Lemma \ref{lem:keepinv}. The path $Q$ (\emph{dashed}) is a residual path from $a \in A$ to $b \in B$, which was created after adding $Q'$ (\emph{solid}). In both cases, when $S \subseteq A$ or $T \subseteq B$, we get that there was a residual path from $A$ to $B$ also before adding $Q'$.}
\end{figure}

\section{The algorithm}

The general framework of our algorithm is Algorithm \ref{alg:LR}, which we described in Section \ref{sec:sum}. This algorithm is recursive, for the base of the recursion we can use the cases where $|S| = O(1)$, $|T| = O(1)$ or $|S||T| = O(\sqrt{n})$, in all of these cases we can find a maximum flow from $S$ to $T$ in $O(n^{3/2} \log n)$ time, as described in the introduction.

We use a procedure that computes a maximum flow between sources that are in one piece of the graph and sinks that are in another piece, with respect to a separator of size $r$. In the next section, we give an algorithm for this task which runs in $O(n^{3/2} \log n + rn)$ time.

We assign to each source or sink in $G$ a weight of $(|S| + |T|)^{-1}$, the other vertices have weight $0$, and find a separator $P$ that separates the graph into two pieces -- $X$ and $Y$, such that the total number of sources and sinks in each of them at most $(2/3)(|S|+|T|)+ O(\sqrt{n})$ (the term $O(\sqrt{n})$ comes from the fact that the vertices of $P$, that are common to $X$ and $Y$, can also be sources or sinks). We split the sources and sinks of $G$ into two disjoint sets $L$ and $R$ (Step 1 of Algorithm \ref{alg:LR}) as follows -- $L = X_s \cup X_t \setminus P_s$, $R = Y_s \cup Y_t \setminus P_t$. In other words, we define $L$ and $R$ according to $X$ and $Y$, and for the vertices of $P$, which are common to $X$ and $Y$, we put the sinks in $L$ and the sources in $R$.

Next (Step 2), we find a maximum flow from $L_s$ to $R_t$. With the algorithm of Section \ref{sec:sidetoside}, this takes $O(n^{3/2} \log n)$ time. After we found the maximum flow $f$, it is easy to find its cut $C$ in linear time. We remove the edges of $C$ from the graph (Step 3).

Now (Step 4), we find a maximum flow inside each connected component. Let $K$ be a component of $G \setminus C$ with $n_K$ vertices. The component $K$ is of one of two types, either $K \cap R_t$ is empty or $K \cap L_s$ is empty.

For the first type we use Lemma \ref{lem:subset} to find a maximum flow from $K_s$ to $K_t$ in three stages. First we find a maximum flow from $K \cap R_s$ (if it is not empty) to $K_t = K \cap L_t$. Now it remains to find a maximum flow from $K \cap L_s$ to $K_t$ in the residual network. We split $K_t$ into two parts, we find a maximum flow from $K \cap L_s$ to $K_t \cap P_t$ in the residual network, and from $K \cap L_s$ to $K_t \setminus P_t$ in the residual network (each time we compute a flow, we add it to $f$). We use the algorithm of Section \ref{sec:sidetoside} to perform the first two maximum flow computations in $O(n_K^{3/2} \log n_K + |P \cap K|n_K)$ time (note that $P \cap K$ remains a cycle separator between the sources and the sinks in these two cases), and we use a recursive application of our main algorithm for the third computation. For the second type of components we use Lemma \ref{lem:subset} is a similar way, we compute maximum flow from $K_s = K \cap R_s$ to $K_t \cap L_t$, then from $K_s \cap P_s$ to $ K_t \cap R_t$ in the residual network, and last from $K_s \setminus P_s$ to $K_t \cap R_t$ in the residual network. Again, the first two maximum flow computations are applications of Section \ref{sec:sidetoside}, and the third is a recursive call. The different connected components are pairwise vertex-disjoint, so the total time in Step 4 for all the applications of the algorithm of Section \ref{sec:sidetoside} is $O(n^{3/2} \log n)$. The sources and sinks in each recursive call are either all contained in $L$ or all contained in $R$, and none of them is in $P$, so there are at most $(2/3)(|S|+|T|)$ sources and sinks in each call.

We restore the edges of $C$ (Step 5), and find a maximum flow from $R_s$ to $L_t$ in the residual network (Step 6), again this takes $O(n^{3/2} \log n)$ time using the algorithm of Section \ref{sec:sidetoside}. This concludes our implementation of Algorithm \ref{alg:LR}.

The correctness of our algorithm is derived from the correctness of Algorithm \ref{alg:LR} \cite{MN95}. We call the algorithm recursively on each connected component of $G \setminus C$, the components are disjoint, and the number of sources and sinks in each component is at most $(2/3)(|S|+|T|)$, therefore the total time bound of our algorithm is $O(n^{3/2} \log^2 n)$.

\begin{theorem}
    There is an $O(n^{3/2} \log^2 n)$ time algorithm that solves the multiple-source multiple-sink maximum flow problem in a planar graph with $n$ vertices. 
\end{theorem}

\section{Maximum flow from one side of a separator to another} \label{sec:sidetoside}

In this section we present a solution for a special case of the multiple-source multiple-sink maximum flow problem, where all the sources are in one piece of the graph, and all the sinks are in another piece. We use this as a procedure in our main algorithm. This procedure is similar to a procedure used by Borradaile and Wulff-Nilsen \cite{BW10}, but here we can use a multiple-source single-sink maximum flow algorithm (\cite{BW10} or \cite{KM10}) as a sub-procedure, to avoid (explicit) recursion.

Let $X$ be the piece of $G$ that contains $S$, and let $Y$ be the piece of $G$ that contains $T$. Let $P= {p_1, \dots, p_r}$ be the separator of $r$ vertices that separates between $X$ and $Y$. Since $P$ is a cycle separator, for every $1 \leq i < r$ it is possible to add an edge between $p_i$ and $p_{i+1}$ without violating the planarity of the graph, we will do so below.

For simplicity we assume that no vertex of $P$ is a source or a sink. If this is not true then we can fix it easily -- assume that $s \in P$ is a source, then we add a vertex $s'$ inside a face adjacent to $s$ on the side of $X$ of the separator, connect $s'$ to $s$ with an edge of capacity $\M$, where $\M$ is the sum of all capacities of all arcs in $G$, and make $s'$ a source instead of $s$. The case where $t \in P$ is a sink is handled symmetrically. Note that this transformation keeps the number of vertices and the number of edges in the graph $O(n)$.

We begin by finding a maximum flow $f_X$ from $S$ to all vertices of $P$, inside the piece $X$. We add to $X$ a super-sink $t$ in the place where $Y$ is in the plane, and connect all the vertices of $P$ to $t$ with arcs of capacity $\M$. Since $P$ is a cycle separator, adding $t$ to $X$ preserves the planarity of $X$. We find a multiple-source single-sink maximum flow from $S$ to $t$, then $f_X$ is the restriction of this flow to $X$. Symmetrically, we find a maximum flow $f_Y$ from all the vertices of $P$ to $T$, inside the piece $Y$.

Let $f = f_X + f_Y$. Since $X$ and $Y$ are edge disjoint, $f$ is a pseudoflow in $G$ (that is, it satisfies the capacity constraint and the antisymmetry constraint). In the pseudoflow $f$ the vertices that are not sources and not sinks that may have non-zero flow excess are only vertices of $P$.

Next we show how to transform $f$ from pseudoflow to a flow. We do so by balancing the flow excess of each vertex of $P$ with non-zero excess, one vertex after the other, while keeping the following invariant:
\begin{invariant} \label{inv:max}
    Let $P' \subseteq P$ be the vertices of $P$ that we did not process yet. There is no residual path with respect to $f$ from a vertex of $S$ to a vertex of $T$, or from a vertex of $S$ to a vertex of $P'$, or from a vertex of $P'$ to a vertex of $T$.
\end{invariant}
Initially the invariant is true, there are no residual paths from $S$ to $P$ because $f_X$ is maximum, and no residual paths from $P$ to $T$ because $f_Y$ is maximum, also any path from $S$ to $T$ must contain a vertex of $P$, so there are no residual paths from $S$ to $T$. When we are done, $P'$ is empty, and all the vertices of $P$ have $0$ flow excess, so $f$ is a flow function. Since there is no residual path with respect to $f$ from $S$ to $T$, $f$ is a maximum flow.

We go over the vertices of $P$, from $p_1$ to $p_r$. Let $p_i$ be the current vertex that we process, $P' = \{p_j \mid j \geq i\}$ is the set of vertices that we did not process yet. If $\exc{p_i} = 0$ we do nothing and proceed to the next vertex.

The second case is when $\exc{p_i} > 0$. In this case, if $i < r$, we add arcs $(p_r, p_{r-1})$, $(p_{r-1}, p_{r-2})$, $\dots$, $(p_{i+2}, p_{i+1})$, each with capacity $\M$. We find a maximum flow $f_i$, whose value is \emph{bounded} by $\exc{p_i}$, from $p_i$ to $p_{i+1}$ in the residual network. We bound the value of the maximum flow by $\exc{p_i}$ by adding a source $s'$ inside the face that is adjacent both to $p_i$ and to $p_{i+1}$ and connecting $s'$ to $p_i$ with an arc of capacity $\exc{p_i}$, then we actually find a maximum flow from $s'$ to $p_{i+1}$. We remove the arcs we added from $p_r$ to $p_{i+1}$ from the graph, and let $f = f + f_i$. This takes $O(n)$ time using the algorithm of Hassin \cite{H81} for maximum flow in a planar graph where the single source and the single sink are on the same face, with the shortest path algorithm of Henzinger et al.~\cite{HKRS97}.

\begin{lemma}
    After adding $f_i$, the pseudoflow $f$ satisfies Invariant \ref{inv:max}.
\end{lemma}
\begin{proof}
    As in the proof of Lemma \ref{lem:keepinv}, we can view the flow $f_i$ as the sum of simple paths that carry flow from $p_i$ to $p_{i+1}$ and simple cycles of flow \cite{FF62}. We computed the flow $f_i$ on a graph that contains $G$ and additional edges between members of $P'$. If we restrict $f_i$ to the original graph $G$, then we can view $f_i$ as the sum of simple paths, each carries flow from one vertex of $P'$ to another vertex of $P'$, and of simple cycles of flow. Therefore, the restriction of $f_i$ to $G$ is a flow function from a subset of $P'$ to another subset of $P'$.

    Before adding $f_i$ to $f$, there was no residual path from $P' \cup S$ to $T$ by Invariant \ref{inv:max}. From Lemma \ref{lem:keepinv} we get that the same is true after adding $f_i$ to $f$. Similarly, before adding $f_i$ to $f$ there was not residual path from  $S$ to $P'$, again from Lemma \ref{lem:keepinv} we get that the same is true also after adding $f_i$ to $f$. We conclude that $f$ satisfies Invariant \ref{inv:max} also after adding $f_i$.
\end{proof}

After we added $f_i$ to $f$, it is possible that $\exc{p_i} = 0$, if this is the case then we are done with $p_i$. Otherwise (or if $i = r$ and we skipped the computation of $f_i$) we return a flow of value $\exc{p_i}$ from $p_i$ to $S$ and to vertices of $P'$ with negative flow excess. We return the flow using a process of Johnson and Venkatesan \cite{JV83}. First, we make the pseudoflow $f$ acyclic, then we send back from $p_i$ a flow of value $\exc{p_i}$ along a reverse topological order of $f$. Due to the antisymmetry constraint, we can view this process as adding to $f$ a flow function $f'_i$ from $p_i$ to a set of vertices $U \subseteq S \cup P'$, each with a negative flow excess in $f$. This takes $O(n)$ time using the algorithm of Kaplan and Nussbaum for canceling cycles of flow \cite{KN10}.

Now the flow excess of $p_i$ is $0$, this finishes the processing of $p_i$ for the case where the flow excess of $p_i$ was initially positive. After this step $P' = \{p_j \mid j > i\}$.

\begin{lemma}
    After returning the excess flow from $p_i$ to $S$, the pseudoflow $f$ satisfies Invariant~\ref{inv:max}.
\end{lemma}
\begin{proof}
    Before returning the excess flow from $p_i$ to $U \subseteq S \cup P'$ by adding $f'_i$ to $f$, there was no residual path from $p_i$ to any other vertex of $P'$, we get this from the maximality of $f_i$, since the arcs with capacity $\M$ which we added between vertices of $P'$ allow to extend any residual path from $p_i$ to any other vertex of $P'$ to a residual path from $p_i$ to $p_{i+1}$. Also before we added $f'_i$ to $f$, there was no residual path from $S$ to $P'$ from Invariant \ref{inv:max}. Therefore, after adding $f'_i$ to $f$ there is still no residual path from $S$ to $P'$ by Lemma \ref{lem:keepinv} (for the set $A$ in the statement of the lemma we used here the set $\{p_i\} \cup S$, note that $P'$ does not contain $p_i$ anymore). In addition, before we added $f'_i$ to $f$, there was no residual path from $\{p_i\} \cup P' \cup S$ to $T$ from Invariant \ref{inv:max}, and this remains true after adding $f'_i$ to $f$ by Lemma \ref{lem:keepinv}.
\end{proof}

The third case of processing $p_i$ is when $\exc{p_i} < 0$. In this case we symmetrically add arcs $(p_{i+1}, p_{i+2}), (p_{i+2}, p_{i+3}), \dots, (p_{r-1}, p_r)$, each with capacity $\M$. We find a maximum flow $f_i$, whose value is bounded by $-\exc{p_i}$, from $p_{i+1}$ to $p_i$ in the residual network, we remove the arcs we added from $p_{i+1}$ to $p_r$ from the graph, and let $f = f + f_i$. Then, if $\exc{p_i}$ remains negative, we send back flow of value $-\exc{p_i}$ from $T$ to $p_i$. This case also keeps Invariant \ref{inv:max}, the proof is symmetric to the previous case.

When we are done, the flow excess in each vertex of $P$ is $0$, so $f$ is a flow function. By Invariant \ref{inv:max}, there is no residual path with respect to $f$ from a vertex of $S$ to a vertex of $T$. Therefore $f$ is a maximum flow.

It takes us $O(n^{3/2} \log n)$ time to find $f_X$ and $f_Y$, and for each vertex of $P$ we spend $O(n)$ time to eliminate any non-zero flow excess it has. Therefore, the total time for the procedure in this section is $O(n^{3/2} \log n + rn)$.

We note that the bottleneck of our algorithm is the multiple-source single-sink maximum flow computation. For the interesting case where $G$ is a grid, it is possible to use techniques similar to these presented here to improve the running time of the multiple-source single-sink algorithms \cite{BW10,KM10} to $O(n^{3/2})$. This will also improve the time bound of our multiple-source multiple-sink maximum flow algorithm to $O(n^{3/2} \log n)$.

\section*{Acknowledgments}
The author would like to thank Haim Kaplan and Christian Wulff-Nilsen for their comments on the paper.

\end{document}